\documentclass[prl,showpacs]{revtex4-1}
\usepackage{setspace}
\usepackage[dvipdfmx]{graphicx}
\usepackage{graphicx}
\usepackage[dvipdfmx]{graphicx}
\usepackage{bm}
\usepackage{amssymb}   
\usepackage{color}
\usepackage{amsmath}
\usepackage{amsmath,amscd}
\usepackage{caption,subcaption}
\usepackage[colorlinks=true,linkcolor=blue]{hyperref}%
\usepackage{amsthm}
\newtheorem{theorem}{Theorem}
\newtheorem{lemma}[theorem]{Lemma}
\newtheorem{definition}[theorem]{Definition}
\newtheorem{corollary}[theorem]{Corollary}

\begin{document}
\title{Universality of the Route to Chaos -Exact Analysis-}%
\author{Ken-ichi Okubo}%
\email{okubo.kenichi.65z@st.kyoto-u.ac.jp}
\affiliation{Department of Applied Mathematics and Physics, 
	Graduate School of Informatics, Kyoto University} %
\author{Ken Umeno} %
\email{umeno.ken.8z@kyoto-u.ac.jp}
\affiliation{Department of Applied Mathematics and Physics, 
	Graduate School of Informatics, Kyoto University} %
\date{\today}%
\begin{abstract}
	The universality of the route to chaos is analytically proven for countably infinite number of maps by
	proposing the Super Generalized Boole (SGB) transformations.
	As one of the route to chaos, intermittency route was studied by Pomeau and Manneville numerically. 
	They conjectured the universality in \textit{Type 1} intermittency, that the critical exponent of the Lyapunov exponent is $1/2$ in \textit{Type 1} intermittency.
	In order to prove their conjecture, we showed that for certain parameter ranges, the SGB transformations
	are \textit{exact} and preserve the Cauchy distribution. Using the property of exactness, we proved that
	the critical exponent is $1/2$ for countably infinite number of maps where \textit{Type 1} intermittency occurs.
	
\end{abstract}
\maketitle

\paragraph{Universality in chaos} 
Route from stable states to chaotic (intermittent) states
has caught much attention in broad fields in physics. 
This issue treats fundamental change of systems from stable state to unstable state
and it is an essential theme to analyze the stability of physical systems. 
Route to chaos is also studied theoretically and experimentally such as Hamiltonian systems \cite{Hioe}, 
map systems \cite{Otto,Manneville,Huberman,Pomeau,Milosavljevic,Lamba}, coupled oscillators \cite{Liu},
Belousov-Zhabotinskii reaction \cite{Swinney}, Rayleigh-Brnard convection \cite{Swinney}, Couette Taylor flow \cite{Swinney},
noise induced system \cite{Crutchfield}, thermoacoustic system \cite{Kabiraj} and optomechanics \cite{Bakemeier,He,Coillet}.
There is the theoretical classification of routes to chaos such as intermittency route, 
period doubling route, frequency locking route, etc \cite{Swinney}.
Frequently these researches have been motivated to discover the universality at the onset of chaos
with respect to the critical exponent of the Lyapunov exponent, which is an indicator of chaos.
The universality of the critical exponents in each route to chaos has been studied extensively by numerical simulations.
For period doubling route, Huberman and Rudnick \cite{Huberman} estimated numerically the critical exponent $\nu$ 
as $\nu = \frac{\log 2}{\log \delta}$ where $\delta$ represents the Feigenbaum constant.
For intermittency route treated in this Letter, Pomeau and Manneville \cite{Manneville, Pomeau} classified intermittency into three types and 
conjectured the universality of $\nu$ for each intermittent type. 
In particular, in \textit{Type 1} intermittency, they conjectured the universality that $\nu = \frac{1}{2}$, by the numerical simulations.
After the work by Pomeau and Manneville, the critical exponent
has been researched in various field with relations to intermittency such as Billiard system \cite{Benettin,Muller}, 
electronic circuit \cite{Ono}, plasma physics \cite{Feng} and intermittent map \cite{Cosenza}. 
Those studies by numerical simulations suggest that their conjecture $\nu=\frac{1}{2}$ would be right.

However, in these researches, the critical exponent $\nu$ is estimated by numerically or
the analytical formulae of the Lyapunov exponent $\lambda$ was not obtained without any assumption.

On the other hand, the present authors \cite{Okubo} led the analytical formula of the Lyapunov exponent $\lambda$
as an explicit  function in terms of the bifurcation parameter by showing the mixing property for the Generalized Boole (GB) transformation. 
They proved in the GB transformation that 
both \textit{Type 1} and \textit{Type 3} intermittency occur
and that the conjecture by Pomeau and Manneville is correct.

In this Letter, it is analytically proved that for countably infinite number of maps, it holds that
\begin{equation}
\lambda \sim b\left|\alpha-\alpha_c\right|^\nu,~\nu = \frac{1}{2}, b>0,\label{Universality}
\end{equation} 
when \textit{Type 1} intermittency occurs where $\alpha$ and $\alpha_c$ represent a bifurcation parameter and the critical point, respectively.
In order to prove this, we propose more generalized maps,
the Super Generalized Boole (SGB) transformations and show that there are parameter ranges in which
the SGB transformations are \textit{exact} (stronger condition than ergodicity). 
That means one obtains countably infinite number of exact (ergodic) maps. 
Using this result, one can obtain explicitly the analytical formulae of the Lyapunov exponents 
and critical exponents.

We define two-parameterized one-dimensional maps, 
the Super Generalized Boole Transformations (SGB) $S_{K, \alpha} : \mathbb{R}\backslash B \to \mathbb{R}\backslash B$ as follows.
\begin{eqnarray}
x_{n+1} = S_{K, \alpha} (x_n) \overset{\mathrm{def}}{=} \alpha K F_K(x_n), \label{Definition: SGB}
\end{eqnarray}
where $\alpha>0$, $K \in \mathbb{N}\backslash \{1\}$, the set $B$ is a set on $\mathbb{R}$ such that for any point $x$ on $B$, there is an integer $n$ where
$S_{K, \alpha}^n x$ reaches a singular point, and the function $F_{K}$ corresponds to $K$-angle formula of cot function defined 
in Supplemental material.
The GB transformation $T_{\alpha, \alpha}$ in \cite{Okubo} corresponds to the map $S_{2, \alpha}$.
Figure \ref{Fig: Form of SGB} shows the return maps of $S_{3, \frac{1}{3}}, S_{4, \frac{1}{4}}$ and $S_{5, \frac{1}{5}}$.
\begin{figure}[!h]
	\centering
	\includegraphics[width=.8\columnwidth]{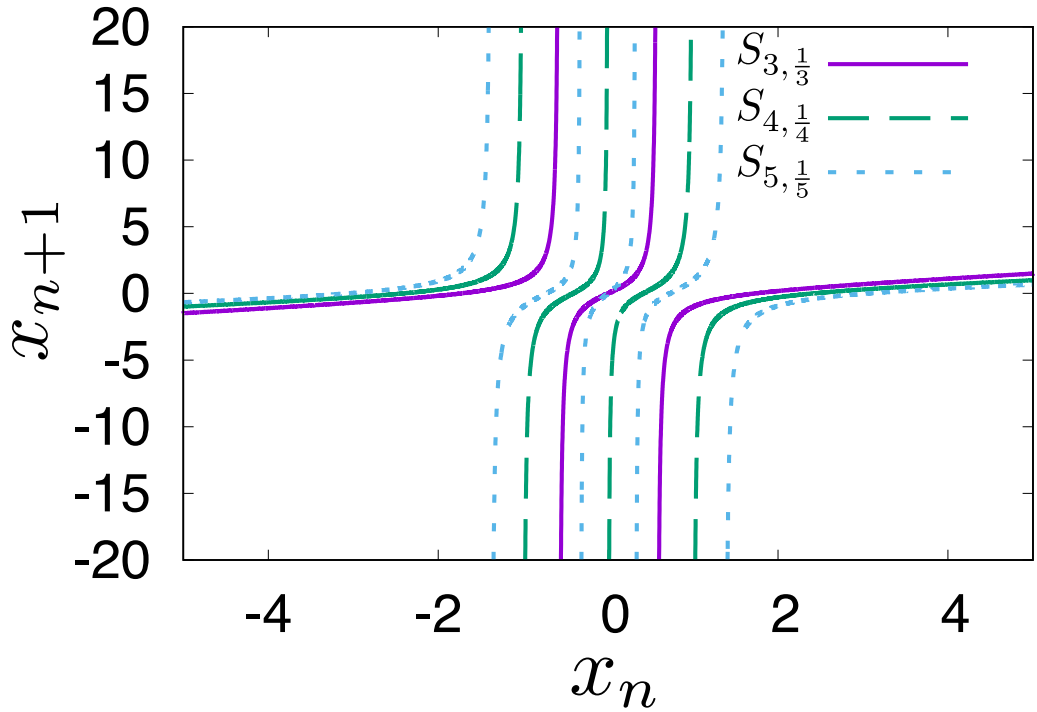}
	\caption{{\footnotesize The return maps of $S_{3, \frac{1}{3}}, S_{4, \frac{1}{4}}$ and $S_{5, \frac{1}{5}}$.
			The solid line, broken line, dotted line correspond to the form of $S_{3, \frac{1}{3}}, S_{4, \frac{1}{4}}$ and $S_{5, \frac{1}{5}}$, respectively.
			The explicit forms of these three maps are in Supplemental material.}}
	\label{Fig: Form of SGB}
\end{figure}

\paragraph{Invariant density} In this paragraph, it is proven that the SGB transformations preserve the Cauchy distribution for certain condition.
According to \cite{Umeno98,Umeno16}, the map $S_{K, \alpha}$ is $K$ to one map as follows.
$y = K\alpha \cot K\theta = K\alpha F_K(x_j),$ $x_j = \cot\left(\theta+j\frac{\pi}{K}\right), j=1,2,\cdots,K$.
If variables $\{x_j\}$ obey the Cauchy distribution $f_\gamma(x) = \frac{1}{\pi}\frac{\gamma}{x^2 + \gamma^2}$ whose scale parameter is $\gamma$, then
according to \cite{Umeno16}, the variable $y$ obeys the density function 
$p(y) = \frac{1}{\pi} \frac{\alpha K G_K(\gamma)}{y^2+ \alpha^2K^2G_K^2(\gamma)}$, 
where the function $G_K(x)$ corresponds to $K$-angle formula of coth function defined in Supplemental material.
Then the  scale parameter $\gamma$ is transformed in one iteration as
$\gamma \longmapsto\alpha K G_K(\gamma)$. 
Now, for each $K$, let us obtain the fixed point $0<\gamma_{K, \alpha}<\infty$ which satisfies the relation
\begin{eqnarray}
\gamma_{K, \alpha} = \alpha K G_K(\gamma_{K, \alpha}) \label{invariant density}
\end{eqnarray}
and clarify the condition of $\alpha$ that there exists a solution of \eqref{invariant density}.
The Cauchy distribution whose scale parameter is a solution of \eqref{invariant density}, which corresponds to the invariant density.
In order to approach this problem, we define the Condition A as follows.
\begin{definition}
	Condition A is referred to as
	\begin{equation}
	\left\lbrace
	\begin{array}{llllll}
	0<\alpha<1 & \mbox{in the case of}& K=2N, \\
	\frac{1}{K^2}< \alpha <1 &\mbox{in the case of}& K=2N+1,
	\end{array}
	\right.
	\end{equation}
	where $N \in \mathbb{N}$.
\end{definition}
Then the following theorem holds. 
\setcounter{theorem}{0}
\renewcommand{\thetheorem}{\Alph{theorem}}
\begin{theorem}\label{Theorem: unique solution}
	When the Condition A is satisfied, the SGB transformations $\{S_{K, \alpha}\}$ preserve the Cauchy distribution and 
	the scale parameter can be chosen uniquely.
\end{theorem}
\setcounter{theorem}{1}
\renewcommand{\thetheorem}{\arabic{theorem}}
The proof of Theorem \ref{Theorem: unique solution} is given in Supplemental material.  
Although it has been proven that the map $S_{K, \alpha}$ preserves the Cauchy distribution and its scale parameter $\gamma_{K,\alpha}$ can be determined uniquely
when the Condition A is satisfied, 
it is not straightforward to obtain the explicit form of fixed point $\gamma_{K, \alpha}$ for arbitrary $K$,
since we have to solve the $K$th-degree equations.
From Theorem \ref*{Theorem: unique solution}, the condition that there exists only 
one solution of \eqref{invariant density} which satisfied $0<\gamma_{K,\alpha}<\infty$
is nothing but the Condition A.

\paragraph{Exactness}
According to \cite{Mackey,Lasota,Schwegler}, the exactness is defined as follows.
\begin{definition}[Exactness]
	A dynamics $T$ on a phase space $\mathcal{X}$ with transfer operator $\mathcal{P}_{T}$ and 
	unique stationary density $f_*$ is called to be exact if and only if
	\begin{equation}
	\lim_{n \to \infty} \|\mathcal{P}_{T}^n f -f_*\|_{L^1} =0,
	\end{equation}
	for every initial density $f \in \mathcal{D}$ where $\mathcal{D}$ denotes all densities on $\mathcal{X}$.
	
	\noindent This definition is equivalent to as follows,
	\begin{equation}
	\lim_{n \to \infty} \mu_*(T^n s) = 1,~ {}^\forall s \in \mathcal{B}, ~\mu_*(s) >0,
	\end{equation}
	where $\mathcal{B}$ denotes the $\sigma$-algebra and $\mu_*$ denotes the invariant measure
	corresponding the invariant density $f_*$.
\end{definition}
In terms of exactness, we obtain the following theorem.
\setcounter{theorem}{1}
\renewcommand{\thetheorem}{\Alph{theorem}}
\begin{theorem}\label{Exact}
	If the the Condition A is satisfied, the SGB transformations $\{S_{K, \alpha}\}$ are exact.
\end{theorem}
\setcounter{theorem}{2}
\renewcommand{\thetheorem}{\arabic{theorem}}
The proof is given in Supplemental material.
From Theorems \ref{Theorem: unique solution} and \ref{Exact}, when the Condition A is satisfied, 
the map $S_{K, \alpha}$ preserves certain Cauchy distribution $f_*$ and 
any initial density function $f$ defined on $\mathbb{R}\backslash B$
converges to $f_*$ as
\begin{equation}
\lim_{n \to \infty} \|\mathcal{P}_{S_{K, \alpha}}^nf - f_*\|_{L^1} =0.
\end{equation}
For example, $S_{3, \alpha}$, $S_{4, \alpha}$ and $S_{5, \alpha}$ are \textit{exact} for $\frac{1}{9}< \alpha <1$,
$0<\alpha <1$ and $\frac{1}{25} < \alpha < 1$, respectively.
According to \cite{Lasota}, if the SGB transformations are exact, then the corresponding dynamical systems are
mixing and ergodic. Therefore the following Corollary holds.
\begin{corollary}
	Suppose that the Condition A is satisfied. Then the dynamical system $(\mathbb{R}\backslash B, S_{K, \alpha}, \mu_*)$ has the mixing property
	and it is ergodic where  $\mu_*$ is the invariant measure corresponding to the invariant density $f_*$.
\end{corollary}

Using the property of exactness, one can obtain the \textit{explicit} formula of the Lyapunov exponent such that
\begin{equation}
\lambda_{K, \alpha} = \frac{1}{\pi}\int_{-\infty}^\infty \log\left|\frac{dS_{K, \alpha}}{dx}\right| \frac{\gamma_{K, \alpha}}{x^2 +\gamma_{K, \alpha}^2}dx.
\label{the Lyapunov exponent}
\end{equation}
From Pesin's formula, one sees that the Kolmogorov-Sinai entropy is equivalent to the Lyapunov exponent since the SGB transformations
are a one-dimensional map.

For $\alpha>1$, changing variable as $z_n = 1/x_n$, one obtains the map $\widetilde{S}_{K, \alpha}$ defined as
\begin{equation}
\left\lbrace
\begin{array}{lll}
\widetilde{S}_{2N, \alpha}(z) &\overset{\mathrm{def}}{=}& 
\frac{\displaystyle \sum_{i=0}^{N-1}(-1)^i {}_{2N}C_{2N-2i-1}z^{2N-2i-1}}{\displaystyle \alpha K\displaystyle \sum_{i=0}^{N}(-1)^i {}_{2N}C_{2N-2i}z^{2N-2i}},\\
\widetilde{S}_{2N+1, \alpha}(z) &\overset{\mathrm{def}}{=}& 
\frac{\displaystyle \sum_{i=0}^N(-1)^i{}_{2N+1}C_{2N-2i+1}z^{2N-2i+1}}{\displaystyle \alpha K\displaystyle \sum_{i=0}^{N}(-1)^i {}_{2N+1}C_{2N-2i} z^{2N-2i}}.
\end{array}
\right.
\end{equation}
Then one has that $\left|\frac{d\widetilde{S}_{K, \alpha}}{dz}(0)\right| = \frac{1}{\alpha}<1$,
so that for any $K$, the orbits are attracted into the infinite point for $\alpha>1$.
Thus, the Lyapunov exponent $\lambda_{K, \alpha}$ for $\alpha>1$ is derived from the inclination at the infinite point as
\begin{equation}
\lambda_{K, \alpha} = \log \alpha,~~\mbox{for}~{}^\forall K \in\mathbb{N}\backslash\{1\}.
\end{equation}

\paragraph{Scaling behavior}
At the edges of the Condition A, one has that 
\begin{equation}
\begin{array}{lll}
\gamma_{K, \alpha}&=&\infty, 
\begin{array}{lll}
\mbox{for} & \alpha=1 & {}^\forall K,
\end{array}\\
\gamma_{K, \alpha} &=&0, \left\lbrace
\begin{array}{llll}
\mbox{for} & \alpha=0 & \mbox{in the case of}& K=2N,\\
\mbox{for} & \alpha=\frac{1}{K^2}& \mbox{in the case of}& K=2N+1.
\end{array}
\right.
\end{array}
\end{equation}
Then the Lyapunov exponent converges to zero at the edges of the Condition A.
In order to discuss the critical phenomena, define critical points as $\alpha_{c1}=1$, $\alpha_{c2} = \frac{1}{(2N+1)^2}$ and $\alpha_{c3}=0$
and define critical exponents $\nu_1$, $\nu_2$ and $\nu_3$ corresponding to $\alpha_{ci}, i=1,2,3$.
In terms of the scaling behavior of the Lyapunov exponents $\lambda \sim b\left|\alpha- \alpha_{ci}\right|^{\nu_i}, b>0, i=1,2,3$
the following theorem holds.
\setcounter{theorem}{2}
\renewcommand{\thetheorem}{\Alph{theorem}}
\begin{theorem}\label{Scaling behavior}
	Suppose that the Condition A is satisfied.
	\begin{itemize}
		\item For any $K\in \mathbb{N}\backslash\{1\}$, it holds that 
		$\nu_1= \frac{1}{2}$ as $\alpha \to 1-0$.
		\item For any $K\in \mathbb{N}\backslash\{1\}$, it holds that 
		$\nu_1= 1$ as $\alpha \to 1+0$.
		\item For $K=2N+1$, it holds that $\nu_2 = \frac{1}{2}$ as $\alpha \to \frac{1}{K^2}+0$.
	\end{itemize}
\end{theorem}
\setcounter{theorem}{2}
\renewcommand{\thetheorem}{\arabic{theorem}}
The proof is given in Supplemental material.

Discuss the Floquet multipliers in the case of $(K, \alpha) = ({}^\forall K, \alpha_{c1})$, $(2N+1, \alpha_{c2})$,
and $(2N, \alpha_{c3})$ for $N \in \mathbb{N}$ .
By changing variable as $x=\cot \theta$, the derivative of the map $S_{K, \alpha}$ is rewritten as
$\frac{dS_{K, \alpha}}{dx} = \alpha K^2 \frac{\sin^2\theta}{\sin^2K\theta}$.

\noindent (i) In the case of $({}^\forall K, \alpha_{c1})$, the derivatives at the infinite point are denoted as
\begin{equation}
\displaystyle \lim_{x \to +\infty} \frac{dS_{K, 1}}{dx} 
= \lim_{\theta \to +0} \frac{(K\theta)^2}{\sin^2 K\theta} \frac{\sin^2\theta}{\theta^2}=1.
\end{equation}
Thus, the Floquet multiplier $\chi$ for $(K, \alpha)=({}^\forall K, \alpha_{c1})$ is unity.

\noindent (ii) In the case of $(2N+1, \alpha_{c2})$, the original point is the fixed point. 
Then, the derivative at the original point is denoted as
\begin{equation}
\displaystyle \frac{dS_{K, \frac{1}{K^2}}}{dx}(0) 
= \frac{\sin^2(\frac{\pi}{2})}{\sin^2\left\lbrace (2N+1)\frac{\pi}{2}\right \rbrace}=1.
\end{equation}
Thus,  at $(K, \alpha) = (2N+1, \alpha_{c2})$, it holds that $\chi=1$.

\noindent (iii) In the case of $(2N, \alpha_{c3})$, applying scale transformation such that $x=\sqrt{\alpha}y$,
one has $y_{n+1} = \widehat{S}_{K, \alpha}(y_n)$ and
$y_{n+1}= \widehat{S}_{K, 0}(y_n) = -\frac{1}{y_n}$. Then, it holds that $\chi = -1$ at $(K, \alpha) = (2N, \alpha_{c3})$.

From (i), (ii) and (iii), one sees that
in the case of $K=2N+1$, only \textit{Type 1} intermittency occurs
at $\alpha=\alpha_{c1}$ and $\alpha_{c2}$
and that in the case of $K=2N$, \textit{Type 1} intermittency occurs at $\alpha=\alpha_{c1}$ and
\textit{Typle 3} intermittency occurs at $\alpha=\alpha_{c3}$.
Therefore, it has been proven that for countably infinite number of \textit{exact} maps, 
the \textit{universal} scaling behavior \eqref{Universality} holds where \textit{Type 1} intermittency occurs.
In Supplemental material, the Floquet multipliers corresponding to $K=3,4$ and 5 are illustrated.

\paragraph{In the case of $K=3, 4$ and 5}
In this paragraph, the examples corresponds to $K=3,4$ and 5 are illustrated.
The solutions of \eqref{invariant density} which satisfies $0<\gamma_{K, \alpha} <\infty$ are uniquely determined as follows.
\begin{equation}
\begin{array}{lll}
\gamma_{3, \alpha} &=& \displaystyle\sqrt{\frac{9\alpha-1}{3- 3\alpha}},\\
\gamma_{4, \alpha} &=& \displaystyle \sqrt{\frac{6\alpha-1+\sqrt{32\alpha^2-8\alpha+1}}{2(1-\alpha)}},\\
\gamma_{5, \alpha} &=& \displaystyle\sqrt{\frac{-5(1-5\alpha)+\sqrt{20(25\alpha^2-6\alpha+1)}}{5(1-\alpha)}}.
\end{array}
\end{equation}
From the above discussion, one knows that in the case of $K=3$ and 5, only \textit{Type 1} intermittency occurs
and in the case of $K=4$, both \textit{Type 1} and \textit{Type 3} intermittency occur.

The Lyapunov exponents in the case of $K= 3, 4$ and 5 are given as follows.
\begin{equation}
\begin{array}{lll}
\lambda_{3, \alpha} &=&  \displaystyle \log \left|\frac{1}{\alpha}\left( \frac{3(1-\alpha)}{8}\right)^2\left[1+\sqrt{\frac{9\alpha-1}{3-3\alpha}}\right]^4 \right|,
\\
\lambda_{4, \alpha} &=& \displaystyle \log\left|\frac{\alpha(1+\gamma_{4, \alpha})^6}{\gamma_{4, \alpha}^2(1+\gamma_{4, \alpha}^2)^2}\right|,\\
\lambda_{5,\alpha} &=&   \log \left|\frac{25}{256\alpha}\frac{(1-\alpha)^4}{(\sqrt{125\alpha^2-30\alpha+5}+11\alpha-1)^2}|1+\gamma_{5, \alpha}|^8\right|.
\end{array} \label{Lyapunov K}
\end{equation}
\begin{figure}[!h]		
	\centering
	\begin{minipage}[t]{.5\columnwidth}
		\includegraphics[width=\columnwidth]{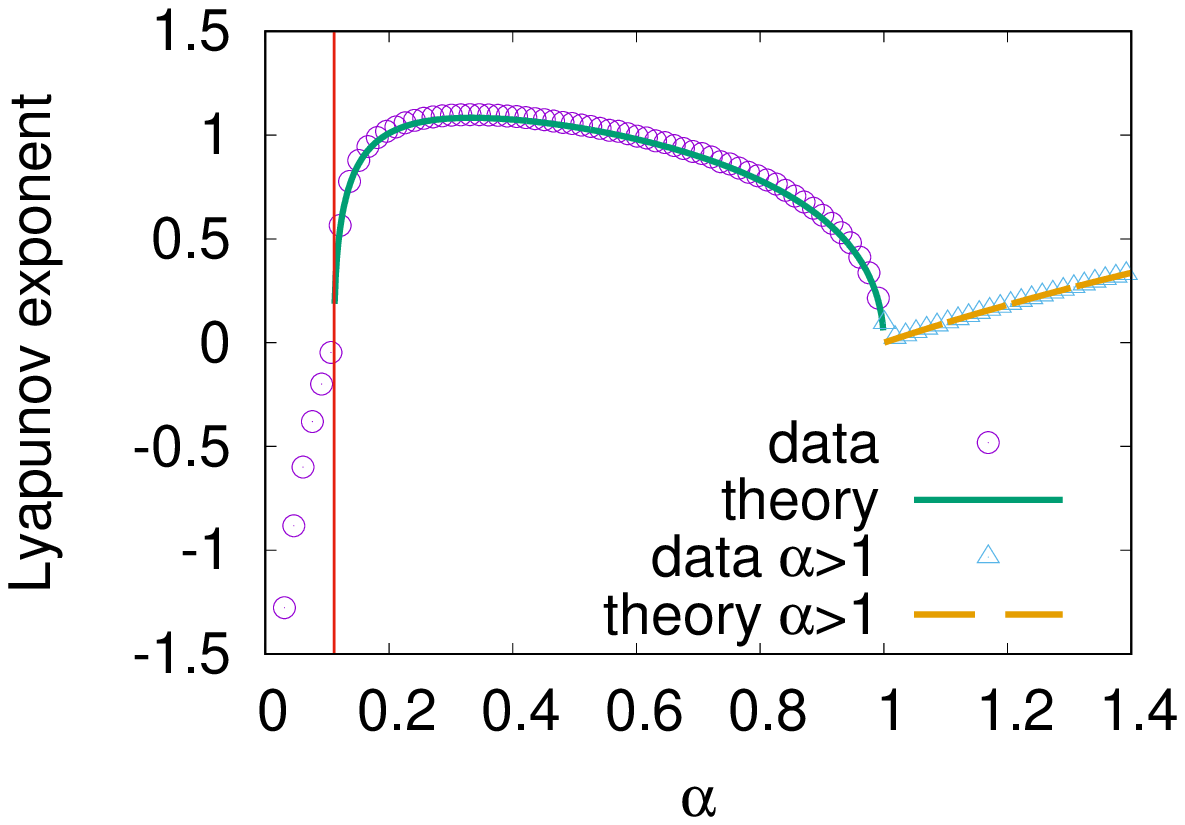}
		
		\subcaption{$K=3$}
		\label{Fig: Lyapunov K=3}    
	\end{minipage}
	
	\begin{minipage}[t]{.5\columnwidth}
		\includegraphics[width=\columnwidth]{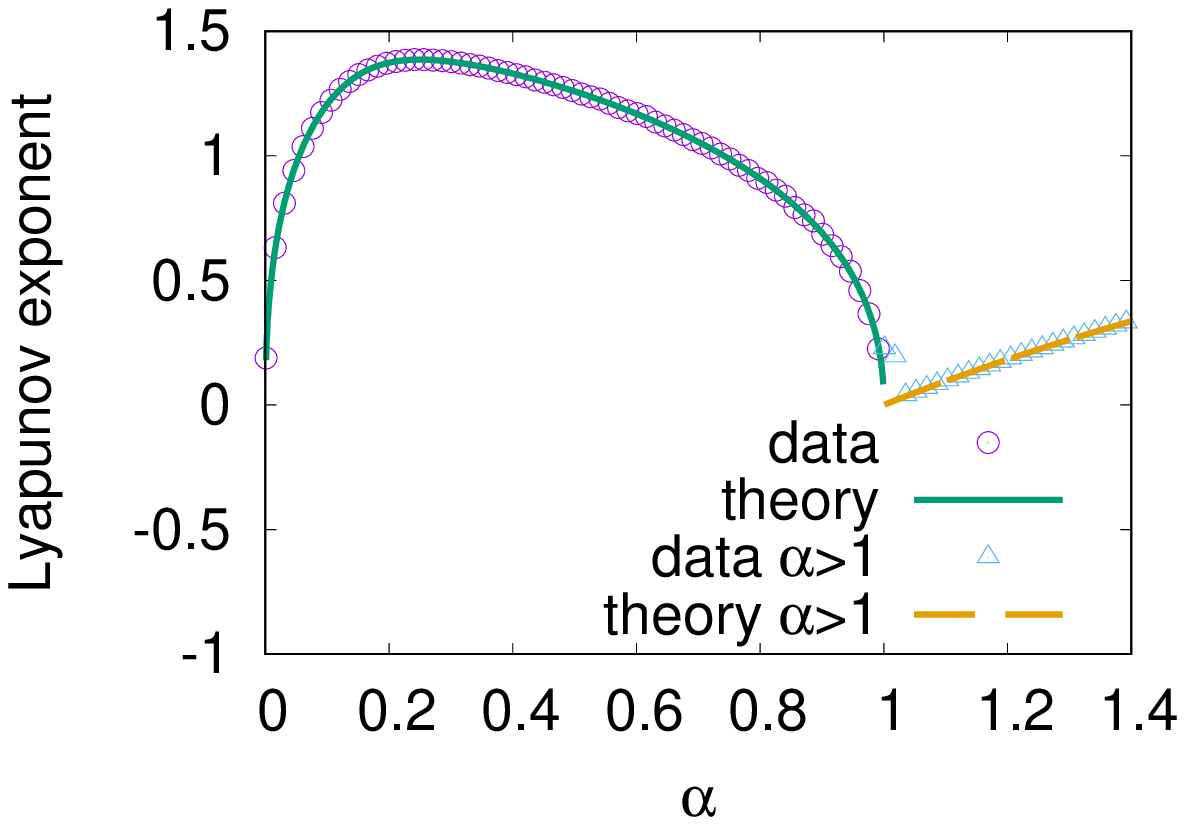}
		
		\subcaption{$K=4$}
		\label{Fig: Lyapunov K=4}
	\end{minipage}
	
	\begin{minipage}[t]{.5\columnwidth}
		\includegraphics[width=\columnwidth]{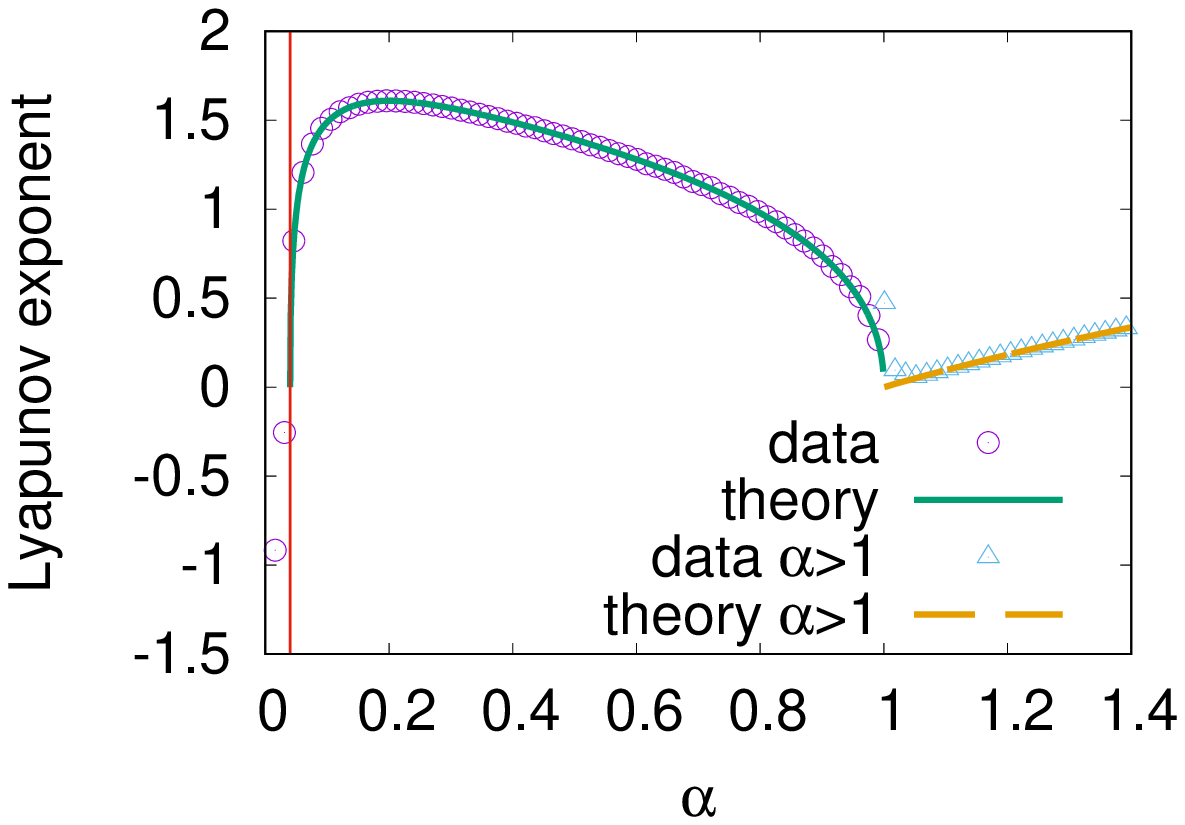}
		
		\subcaption{$K=5$}
		\label{Fig: Lyapunov K=5}
	\end{minipage}
	\caption{{\footnotesize 
			Relations between the Lyapunov exponents of the SGB transformations and $\alpha$ for $K=3,4$ and 5.
			Circles and triangles represent numerical results for $\alpha\leq1$ and for $\alpha>1$, respectively.	  
			The solid lines and broken lines represent the analytical results for $\alpha\leq1$ and for $\alpha>1$, respectively.	
			The initial point is $x_0 = 5\sqrt{7}$. The iteration number is $1\times 10^5$ for $\alpha \leq 1$
			and 200 for $\alpha>1$.
			A vertical line corresponds to $\alpha=\frac{1}{9}, \frac{1}{25}$,
			respectively.}}
	\label{Fig: Lyapunov}
\end{figure}

\begin{figure}[h]		
	\centering
	\includegraphics[width=\columnwidth]{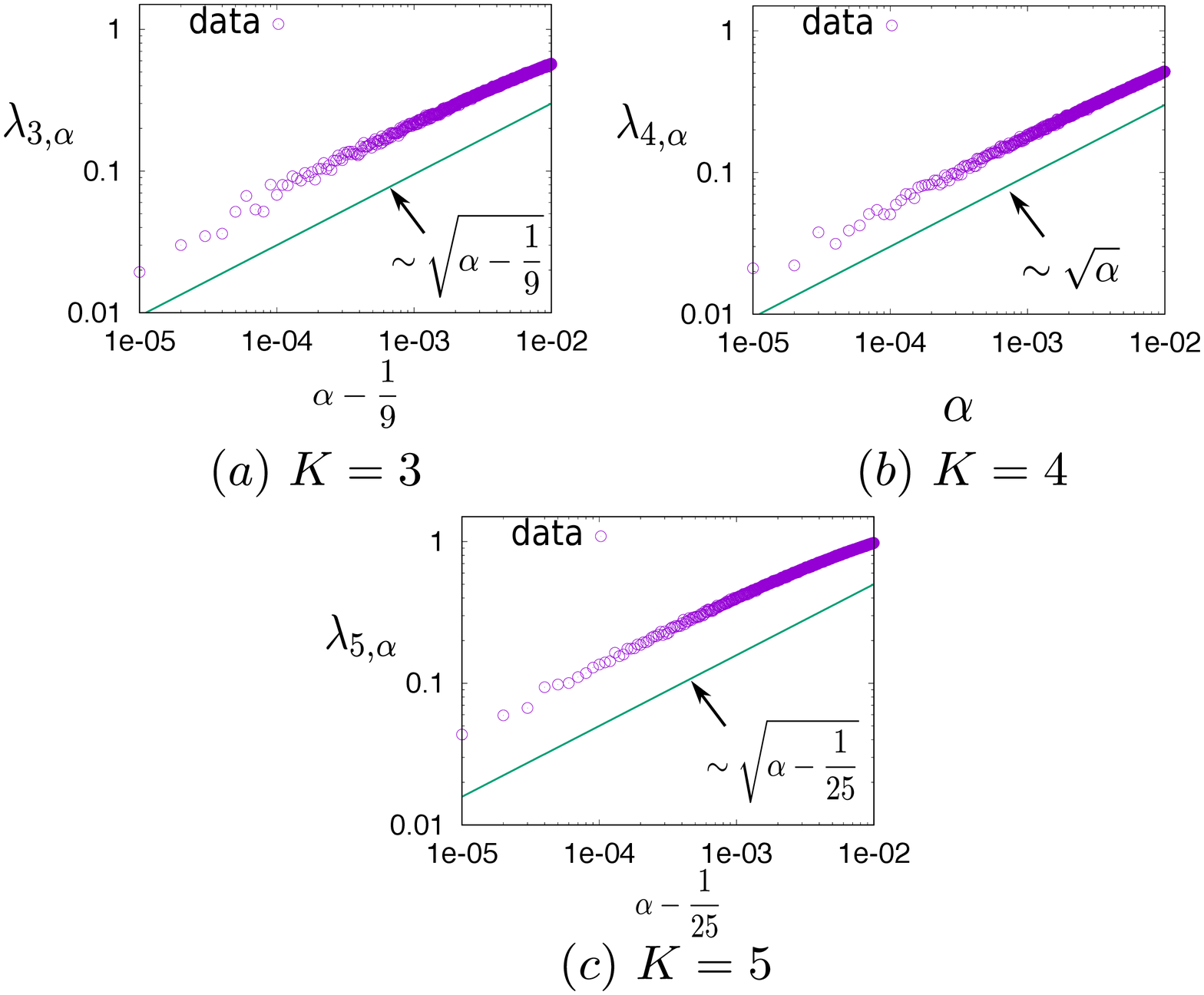}
	\caption{{\footnotesize 
			Scaling behavior of the Lyapunov exponents of the SGB transformation in the case of $K=3,4$ and 5.
			Circles represent numerical simulation and solid lines represent the order.
			The initial point is $x_0 = 5\sqrt{7}$. The iteration number is $2\times 10^5$.}}
	\label{Fig: Lyapunov scaling}
\end{figure}
Figures \ref{Fig: Lyapunov K=3}, \ref{Fig: Lyapunov K=4} and \ref{Fig: Lyapunov K=5}
show the Lyapunov exponents against $\alpha$ in the case of $K = 3,4$ and 5, respectively.
One sees that the numerical simulations are exactly consistent with the obtained analytical formulae.
Since it holds that at the critical points $\partial \lambda_{K, \alpha}/\partial \alpha = \pm \infty$,
one sees that the parameter dependence of the Lyapunov exponent diverges at the critical points.
This means the computational difficulty in obtaining the true value of the Lyapunov exponent by numerical simulation.
Figure \ref{Fig: Lyapunov scaling} shows the scaling behavior of the Lyapunov exponent. 
One sees that $\nu_2=\frac{1}{2}$ and $\nu_3=\frac{1}{2}$ in the case of $K=3,4$ and 5.  

\paragraph{Conclusion}
This work is the first example in which the conjecture by Pomeau and Manneville expressed in \eqref{Universality} 
is analytically proven to be true for countably infinite number of maps (the proposed Super Generalized Boole transformations).
This work shows the theoretical picture of stable-unstable transition for intermittent maps.
In the course of proof, we have shown that 
the Super Generalized Boole (SGB) transformations preserve the unique Cuachy distribution,
together with proving the fact that SGB transformations are \textit{exact} and 
that any initial density function converges to the invariant Cauchy distribution when the Condition A is satisfied.

Applying the property of exactness, one can obtain \textit{analytical} formulae of the Lyapunov exponents for the SGB transformations. 
In the SGB transformations, the Lyapunov exponents $\lambda_{K, \alpha}$ are equivalent to 
the Kolmogorov-Sinai entropy applying to the Pesin's theorem.
Using the analytical formulae of the Lyapunov exponents, we have confirmed that for $K=3,4$ and 5,
the derivative $\partial \lambda_{K, \alpha} /\partial \alpha$ diverge at the
critical points and we obtained $\nu_1=\nu_2=\nu_3 = \frac{1}{2}$.
Thus, we have proven the universality of the route to chaos for a large class of the chaotic systems.

As future works, clarifying the scaling relation between the critical exponent $\nu$ and the other critical exponents,
we can obtain a new perspective of chaos in physics.

\clearpage
\begin{acknowledgments}
	Ken-ichi Okubo acknowledges the support of
	Grant-in-Aid for JSPS Research Fellow Grant Number  JP17J07694.
\end{acknowledgments}

\clearpage
\appendix
\setcounter{equation}{0}
\def\theequation{S\arabic{equation}}
\begin{center}
	{\Large \textbf{Supplemental Material}}
\end{center}
In this Supplemental material it is shown that
\begin{enumerate}
	\item the Super Generalized Boole transformations preserve the Cauchy distribution for certain parameter ranges and the Cauchy distribution is determined uniquely,
	\item the SGB transformations are exact for the parameter ranges, and
	\item the critical exponent of Lyapunov exponent is $\nu=\frac{1}{2}$ 
	for all $K \in \mathbb{N}\backslash\{1\}$ when \textit{Type 1} intermittency occurs.
\end{enumerate}
The definitions of $F_K$,  $G_K$, $S_{K, \alpha}$ are written in the following.

\begin{definition}
	The map $F_{K}: \mathbb{R}\backslash A\to \mathbb{R}\backslash A$ is referred to as
	\begin{eqnarray}
	F_K(\cot\theta) \overset{\mathrm{def}}{=} \cot K\theta,
	\end{eqnarray}	
	where $K \in \mathbb{N}\backslash \{1\}$ and the set $A$ is a set on $\mathbb{R}$ such that for any point $x\in A$, $F_K(x)$ reaches a singular point.
\end{definition}

\begin{definition}
	The map $G_K : \mathbb{R} \to \mathbb{R}$ is referred to as
	\begin{eqnarray}
	G_K(\coth \theta) \overset{\mathrm{def}}{=} \coth K\theta,
	\end{eqnarray}
	where $K \in \mathbb{N}\backslash \{1\}$.
\end{definition}

\begin{definition}[Super Generalized Boole Transformation]
	Super Generalized Boole Transformation $S_{K, \alpha} : \mathbb{R}\backslash B \to \mathbb{R}\backslash B$ is referred to as
	\begin{eqnarray}
	x_{n+1} = S_{K, \alpha} (x_n) \overset{\mathrm{def}}{=} \alpha K F_K(x_n), \label{SDefinition: SGB}
	\end{eqnarray}
	where $\alpha>0$, $K\in \mathbb{N}\backslash\{1\}$ and the set $B$ is a set on $\mathbb{R}$ such that for any point $x$ on $B$, there is an integer $n$ where
	$S_{K, \alpha}^n x$ reaches a singular point.
\end{definition}
For example, $S_{3, \alpha}$, $S_{4, \alpha}$ and $S_{5, \alpha}$ are as follows.
\begin{eqnarray}
S_{3,\alpha}(x_n) &=& 3\alpha \frac{x_n^3 - 3x_n}{3x_n^2 -1}, \label{S_{3}}\\
S_{4, \alpha}(x_n) &=& 4\alpha \frac{x_n^4 -6x_n^2 +1}{4x_n^3-4x_n}, \label{S_{4}}\\
S_{5, \alpha}(x_n) &=& 5\alpha \frac{x_n^5-10x_n^3 +5x_n}{5x_n^4-10x_n^2 +1} \label{S_{5}}.
\end{eqnarray}
The derivative of $S_{K, \alpha}$ with respect to $x$ is denoted as
\begin{eqnarray*}
	S_{2N, \alpha}'(x) &=& (2N)^2\alpha \frac{(1+x^2)^{2N-1}}{\left[\displaystyle\sum_{r=0}^{N-1}(-1)^r {}_{2N} \mathrm{C} _{2r+1} x^{2N-2r-1}\right]^2}>0,\\
	S_{2N+1, \alpha}'(x) &=& (2N+1)^2\alpha \frac{(1+x^2)^{2N}}{\left[\displaystyle\sum_{r=0}^{N}(-1)^r {}_{2N+1}\mathrm{C}_{2r+1} x^{2(N-r)}\right]^2}>0.
\end{eqnarray*}

\section{Existence of the invariant density function}
In this section, we prove that
\begin{itemize}
	\item when $K=2N, N \in \mathbb{N}$ and $0<\alpha<1$, $\Longrightarrow$ the map $S_{2N, \alpha}$ preserves the Cauchy distribution and
	it can be determined uniquely, and
	\item when $K=2N+1, N \in \mathbb{N}$ and $\frac{1}{(2N+1)^2}<\alpha<1$, $\Longrightarrow$ the map $S_{2N+1, \alpha}$ preserves the Cauchy distribution and it can be determined uniquely.
\end{itemize}

According to \cite{Umeno16}, the map $S_{K, \alpha}$ is $K$ to one map as follows.
\begin{eqnarray*}
	y &=& K\alpha \cot K\theta = K\alpha F_K(x_j),  j=1,2,\cdots,K,\\
	x_j &=& \cot\left(\theta+j\frac{\pi}{K}\right), j=1,2,\cdots,K.
\end{eqnarray*}
If variables $\{x_j\}$ obey the Cauchy distribution 
\begin{equation*}
f(x) = \frac{1}{\pi}\frac{\gamma}{x^2 + \gamma^2},
\end{equation*}
then according to \cite{Umeno16} it holds that
\begin{eqnarray}
p(y) |dy|&=& f(x_1)|dx_1|+ f(x_2)|dx_2| + \cdots +  f(x_K)|dx_K|,\nonumber\\
p(y) &=& \frac{1}{\pi} \frac{\alpha K G_K(\gamma)}{\alpha^2K^2G_K^2(\gamma)+y^2}.
\end{eqnarray}
Then the  scale parameter $\gamma$ is transformed in one iteration as
\begin{eqnarray}
\gamma \longmapsto\alpha K G_K(\gamma). \label{scaling parameter equation}
\end{eqnarray}

\noindent  Now, for each $K$ let us obtain the fixed point $\gamma_{K, \alpha}$ which satisfies 
\begin{eqnarray}
\gamma_{K, \alpha} = \alpha K G_K(\gamma_{K, \alpha}). \label{Sinvariant density}
\end{eqnarray}
If we discover a real and positive solution of \eqref{Sinvariant density}, it is the evidence that the map $S_{K, \alpha}$
preserves the Cauchy distribution.

The map $G_{K}(x)$ is denoted as
\begin{eqnarray}
G_{2N}(x) &=&   \frac{\displaystyle\sum_{k=0}^{n} {}_{2N}C_{2k} x^{2(N-k)}}{\displaystyle\sum_{k=0}^{N-1} {}_{2N}C_{2k+1} x^{2(N-k)-1}},\label{G_2N}\\
G_{2N+1}(x) &=& \frac{\displaystyle\sum_{k=0}^{n} {}_{2N+1}C_{2k} x^{2(N-k)+1}}{\displaystyle\sum_{k=0}^{N} {}_{2N+1}C_{2k+1} x^{2(N-k)}} \label{G_2N+1}.
\end{eqnarray}
Since at $\alpha=\frac{1}{K}$, the SGB transformation $S_{K, \alpha}$ is equivalent to
$K$-angle formula of cot function, it is obvious that the fixed point of scaling parameter $\gamma_{K, \frac{1}{K}}$ is unity by simple calculation.
In the following, discuss the case that $\alpha \neq \frac{1}{K}$.
Such lemmas hold.

\begin{lemma}\label{lemma: K, unique alpha>1/K}
	For $\frac{1}{K}< \alpha <1$, fix $\alpha$ and there is only one solution which satisfies \eqref{Sinvariant density} in the range of $\gamma_{K,\alpha}>1$ and
	for $0<\alpha< \frac{1}{K}$ there is no solution in the range of $\gamma_{K, \alpha} >1$.
\end{lemma}

\begin{proof}
	In \eqref{Sinvariant density}, change the variable from $\gamma_{K, \alpha}>1$ into $\coth y$. One has that
	\begin{equation}
	\coth y =\alpha K\coth(Ky).
	\end{equation} 
	A function $f(y)$ is defined to be
	\begin{eqnarray}
	f(y) = \coth y -\alpha K\coth(Ky).
	\end{eqnarray}
	In the range where $y\geq 0$, since a function $\coth y$ decreases monotonically, it holds that
	\begin{equation}
	\coth y > \coth(Ky).
	\end{equation}
	Then, if the condition $\frac{1}{K} < \alpha <1$ is satisfies, one has that
	\begin{eqnarray}
	f'(y) = (1-\coth^2 y) - \alpha K^2\left\lbrace 1-\coth^2(Ky)\right\rbrace<0 \label{F decrease}
	\end{eqnarray}
	Thus, the function $f(y)$ decreases monotonically. 
	Let us discuss the value at $y=0$. In the limit of $y \to +0$, it holds that from \eqref{G_2N+1},
	\begin{equation*}
	\lim_{y \to +0} \frac{\coth\left\lbrace Ky\right\rbrace }{\coth y} =
	\left\lbrace
	\begin{array}{lll}
	\displaystyle \lim_{y \to +0} \frac{{}_{2N}C_{0}\coth^{2N}y}{{}_{2N}C_{1}\coth^{2N}y} &=& \frac{1}{2N}, K=2N,\\
	\displaystyle \lim_{y \to +0} \frac{{}_{2N+1}C_{0}\coth^{2N}y}{{}_{2N+1}C_{1}\coth^{2N}y} 	&=& \frac{1}{2N+1}, K=2N+1.
	\end{array}
	\right.
	\end{equation*}
	Thus, one has that
	\begin{equation}
	\lim_{y \to +0} f(y) =
	\left\lbrace
	\begin{array}{lll}
	\displaystyle\lim_{y \to +0}\coth y \left(1-\alpha \frac{K}{2N}\right) &=&+\infty, K=2N,\\
	\displaystyle \lim_{y \to +0}\coth y \left(1-\alpha \frac{K}{2N+1}\right)&=&+\infty, K=2N+1
	\end{array}\label{F y =0}
	\right.
	\end{equation}
	One also has that
	\begin{eqnarray}
	\lim_{y \to \infty} f(y) &=& 1-\alpha K<0. \label{F y = infty}
	\end{eqnarray}
	Thus, from \eqref{F decrease}, \eqref{F y =0} and \eqref{F y = infty}, it is proven that there is only one solution that satisfies
	$f(y)=0$.
	Therefore, for $\frac{1}{K}<\alpha<1$, there is a solution which satisfies \eqref{Sinvariant density} in the rage of $\gamma_{K, \alpha}>1$.
	
	In the case of $0<\alpha< \frac{1}{K}$, it holds that for any $y > 0$
	\begin{equation}
	f(y)  > 0.
	\end{equation}
	Then, there is no solution which satisfies $\gamma_{K, \alpha}>1$.
\end{proof}

\begin{lemma}\label{lemma: K, unique K=2N alpha<1/K}
	In the case of $K=2N$ for $0< \alpha <\frac{1}{K}$, fix $\alpha$. There is only one solution which satisfies \eqref{Sinvariant density} in the range of 
	$0<\gamma_{K,\alpha}<1$ and for $\frac{1}{K} < \alpha$ there is no solution in the range of $0<\gamma_{K, \alpha}<1$.
\end{lemma}

\begin{proof}
	In \eqref{Sinvariant density}, change the variable from $0<\gamma_{K, \alpha}<1$ into $\tanh y$. One has that
	\begin{equation}
	\tanh y = \frac{\alpha (2N)}{\tanh(2Ny)}.
	\end{equation}
	A function $h_{2N}(y)$ is defined to be
	\begin{equation}
	h_{2N}(y) = \tanh y - \frac{\alpha (2N)}{\tanh(2Ny)}.
	\end{equation}
	The derivative of $h_{2N}(y)$ is denoted as
	\begin{equation}
	h_{2N}'(y) = 1-\tanh^2 y + \alpha (2N)^2 \frac{1-\tanh^2(2Ny)}{\tanh^2(2Ny)}>0 \label{H increase}.
	\end{equation}
	One has that
	\begin{equation}
	\begin{array}{lll}
	h_{2N}(0) &=& -\infty<0,\\
	h_{2N}(\infty) &=& 1-\alpha K >0,~\mbox{for}~0<\alpha<\frac{1}{K}. 
	\end{array} \label{H y=0 infinity}
	\end{equation}
	From \eqref{H increase} and \eqref{H y=0 infinity}, one sees that in the case of $K=2N$ for fixed $\alpha$
	which satisfied $0<\alpha < \frac{1}{K}$, there is only one solution 
	which satisfies \eqref{Sinvariant density} in the range of 
	$0<\gamma_{K,\alpha}<1$. 
	
	In the case of $\frac{1}{2N} < \alpha$, it holds that for all $y>0$
	\begin{equation}
	h_{2N}(y) < 0.
	\end{equation}
	Therefore, there is no solution which satisfies $0<\gamma_{K, \alpha}<1$.
\end{proof}

\begin{lemma}\label{lemma: K, unique K=2N+1 alpha<1/K}
	In the case of $K=2N+1$ for $\frac{1}{K^2}< \alpha <\frac{1}{K}$, fix $\alpha$. There is only one solution which satisfies \eqref{Sinvariant density} in the range of 
	$0<\gamma_{K,\alpha}<1$ and for $\frac{1}{K}< \alpha$, there is no solution in the range of $0<\gamma_{K, \alpha}<1$.
\end{lemma}
\begin{proof}
	In \eqref{Sinvariant density}, change the variable from $0<\gamma_{K, \alpha}<1$ into $\tanh y$. One has that
	\begin{equation}
	\tanh y = \alpha(2N+1)^2\tanh\left\lbrace (2N+1)y\right\rbrace.
	\end{equation}
	A function $h_{2N+1}(y)$ is defined to be
	\begin{equation}
	h_{2N+1}(y) = \tanh y - \alpha(2N+1)\tanh\left\lbrace (2N+1)y\right\rbrace.
	\end{equation}
	It holds that
	\begin{equation}
	\begin{array}{lll}
	h_{2N+1}(0) &=& 0,\\
	h_{2N+1}(\infty) &=& 1- \alpha (2N+1)>0, ~\mbox{for}~ \frac{1}{(2N+1)^2} < \alpha < \frac{1}{(2N+1)}.
	\end{array}
	\end{equation}
	The derivative of $h_{2N+1}(y)$ is 
	\begin{equation}
	\begin{array}{lll}
	h_{2N+1}'(y) &=& 1-\alpha(2N+1)^2 + \alpha(2N+1)^2\tanh^2\left\lbrace (2N+1)y\right\rbrace-\tanh^2y,\\
	h_{2N+1}'(0) &=& 1-\alpha(2N+1)^2 < 0,~\mbox{for}~ \frac{1}{(2N+1)^2}<\alpha< \frac{1}{2N+1}.
	\end{array}
	\end{equation}
	The derivative $h_{2N+1}'(y)$ is also expressed using $\sinh$ functions as follows
	\begin{equation}
	h_{2N+1}'(y) = \frac{\sinh^2\left\lbrace(2N+1)y\right\rbrace\left[1-\alpha(2N+1)\frac{\sinh^2y}{\sinh^2\left\lbrace(2N+1)y\right\rbrace}\right]}{\sinh^2y\sinh^2\left\lbrace(2N+1)y\right\rbrace}.
	\end{equation}
	The function $J(y)=\frac{\sinh y}{\sinh\left\lbrace(2N+1)y\right\rbrace}$ decreases monotonously since for the derivative
	\begin{eqnarray}
	J'(y) &=& \frac{
		-\sinh(2Ny) -n \left[\sinh \left\lbrace (2N+2)y\right\rbrace -\sinh(2Ny)\right]
	}{
	\sinh^2\left\lbrace(2N+1)y\right\rbrace
},
\end{eqnarray}
the numerator is denoted as
\begin{eqnarray}
n \sinh\left\lbrace (2N+2)y\right\rbrace ) \left[\frac{n-1}{n}\frac{\sinh(2Ny)}{\sinh(2N+2)y}-1\right]<0.
\end{eqnarray}
Then $J'(y)<0$.
Considering the fact that 
\begin{equation}
\lim_{y \to  \infty} \frac{\sinh(y)}{\sinh(2N+1)y}  = 0,
\end{equation}
a part of $h_{2N+1}'(y)$, $\left[1-\alpha(2N+1)\frac{\sinh^2y}{\sinh^2\left\lbrace(2N+1)y\right\rbrace}\right]$ increases monotonously and
there is a unique point $y_*$ at which the sign of $h_{2N+1}'(y)$ changes from minus to plus.
Therefore, there is a unique point 
$0<y_{**}<\infty$ at which it holds that $h_{2N+1}(y_{**})=0$. 

\noindent From above discussion, one sees that in the case of $K=2N+1$ and $\frac{1}{K^2}<\alpha<\frac{1}{K}$, there is only one solution 
which satisfies \eqref{Sinvariant density} in the range of 
$0<\gamma_{K,\alpha}<1$. 

In the case of $\frac{1}{K}< \alpha$, since it holds that for all $y>0$,
\begin{equation}
h_{2N+1}(y) <0,
\end{equation}
there is no solution in the range of $0<\gamma_{K, \alpha} < 1$.
\end{proof}

From the Lemmas 4, 5 and 6, such lemmas hold.
\begin{lemma}\label{Lemma: K=2N}
	Consider the case of $K=2N$.
	
	\noindent For $\frac{1}{K} \leq \alpha<1$, there is a unique solution of \eqref{Sinvariant density} and
	the solution $\gamma_{K, \alpha}$ is in the range of $\gamma_{K, \alpha}\geq1$.
	For $0 < \alpha < \frac{1}{K}$, there is a unique solution of \eqref{Sinvariant density} and 
	the solution $\gamma_{K, \alpha}$ is in the range of $0<\gamma_{K, \alpha}<1$.
\end{lemma}
\begin{lemma}\label{Lemma: K=2N+1}
	Consider the case of $K=2N+1$.
	
	\noindent For $\frac{1}{K} \leq \alpha<1$, there is a unique solution of \eqref{Sinvariant density} and
	the solution $\gamma_{K, \alpha}$ is in the range of $\gamma_{K, \alpha}\geq1$.
	For $\frac{1}{K^2} < \alpha <\frac{1}{K}$, there is a unique solution of \eqref{Sinvariant density} and 
	the solution $\gamma_{K, \alpha}$ is in the range of $0<\gamma_{K, \alpha}<1$.
\end{lemma}

The Condition A is defined as follows
\begin{definition}
	Condition A is referred to as
	\begin{equation}
	\mbox{Condition A} : ~\left\lbrace
	\begin{array}{clcclr}
	\mbox{in the case of}& K=2N, &\alpha &\mbox{satisfies} & 0<\alpha<1 &\mbox{and}\\
	\mbox{in the case of}& K=2N+1, &\alpha &\mbox{satisfies} &  \frac{1}{K^2}< \alpha <1.
	\end{array}
	\right.
	\end{equation}
\end{definition}
From Lemmas \ref{Lemma: K=2N} and \ref{Lemma: K=2N+1}, such theorem holds.
\setcounter{theorem}{0}
\renewcommand{\thetheorem}{\Alph{theorem}}
\begin{theorem}
	When the Condition A is satisfied, the SGB transformations $\{S_{K, \alpha}\}$ preserve the Cauchy distribution and 
	the scale parameter can be chosen uniquely.
\end{theorem}
\setcounter{theorem}{9}
\renewcommand{\thetheorem}{\arabic{theorem}}

\section{Exactness}
\noindent According to \cite{Mackey,Lasota,Schwegler}, the exactness is defined as follows.
\begin{definition}[Exactness]
	A dynamics $T$ on a phase space $\mathcal{X}$ with transfer operator $\mathcal{P}_{T}$ and 
	unique stationary density $f_*$ is called to be exact if and only if
	\begin{equation}
	\lim_{n \to \infty} \|\mathcal{P}_{T}^n f -f_*\|_{L^1} =0
	\end{equation}
	for every initial density $f \in \mathcal{D}$ where $\mathcal{D}$ denotes all densities on $\mathcal{X}$.
	
	\noindent This definition is equivalent to as follows,
	\begin{equation}
	\lim_{n \to \infty} \mu_*(T^n s) = 1,~ {}^\forall s \in \mathcal{B}, ~\mu_*(s) >0.
	\end{equation}
	where $\mathcal{B}$ denotes the $\sigma$-algebra and $\mu_*$ denotes the invariant measure
	corresponding the invariant density $f_*$.
\end{definition}
\begin{theorem}\label{SExact}
	If the the Condition A is satisfied, the SGB transformations $\{S_{K, \alpha}\}$ are exact.
\end{theorem}
\begin{proof}
	This proof is based on \cite{Okubo}. For the map $S_{K,\alpha}$ defined by \eqref{SDefinition: SGB}, 
	substituting $\cot(\pi\theta_n)$
	into $x_n$, one has the map $\bar{S}_{K, \alpha} : [0, 1) \to [0, 1)$ such that
	\begin{eqnarray}
	\cot(\pi\theta_{n+1}) &=& \alpha K \cot(\pi K\theta_n),\nonumber\\
	\theta_{n+1} &=& \bar{S}_{K, \alpha} (\theta_n) = \frac{1}{\pi} \cot^{-1}\left\lbrace \alpha K \cot(\pi K\theta_n)\right\rbrace.
	\end{eqnarray}
	
	\noindent The derivative of $\bar{S}_{K, \alpha}$ with respect to $\theta$ is as follows.
	\begin{equation}
	\bar{S}_{K, \alpha}'(\theta) = \frac{\alpha K^2\left\lbrace 1+\cot^2(\pi K \theta )\right\rbrace}{\alpha^2K^2\cot^2(\pi K \theta) + 1}>0 ~~\mbox{for}~~ 0<\alpha<1.
	\end{equation}
	Then, $\bar{S}_{K, \alpha}$ increases monotonously. Since it holds that
	\begin{equation}
	\frac{1}{\pi} \cot^{-1}\left\lbrace \alpha K \cot(\pi K\theta_n)\right\rbrace
	= \frac{1}{\pi} \cot^{-1}\left[ \alpha K \cot\left\lbrace \pi K\left(\theta_n + \frac{j}{K}\right)\right\rbrace \right],~j=1,2\cdots, K-1,
	\end{equation}
	form of $\bar{S}_{K, \alpha}$ has the translational symmetry and it can be constructed by
	shifting the form on $[0, \frac{1}{K})$. That is, the map $\bar{S}_{K, \alpha}$ is also $K$ points to one points map and 
	on any interval $I_{j, 1}$, the form of the map $\bar{S}_{K,\alpha}$ is the same as
	that on the interval $I_{0, 1}$.
	Then by operating $\bar{S}_{K, \alpha}^{-1}$, the measure on $[0, 1)$ is 
	divided into $K$ equivalently. We obtain intervals $\{I_{j,n}\}$ defined below by operating $\bar{S}_{K, \alpha}^{-n}$
	into $[0, 1)$. The interval $I_{j, n} \subset [0, 1)$ is defined to be
	\begin{equation}
	\begin{array}{ccl}
	I_{j, k} &\overset{\mathrm{def}}{=}& \left[\eta_{j, n}, \eta_{j+1, n}\right),~ \eta_{j, n} < \eta_{j+1,n}, 0 \leq j \leq K^n-1,\\
	\eta_{0, n} &=& 0~~\mbox{and}~~\eta_{K^n,n}=1,\\
	\bar{S}_{K, \alpha}^n (I_{j, n}) &=& [0, 1),\\
	\mu(I_{j, n}) &=& \frac{1}{K^n}.
	\end{array}
	\end{equation}
	
	For any non zero measure subset $C \subset [0, 1)$, the set $C$ includes cylinder sets $\displaystyle \bigcup_{j,n'}I_{j,n'}$.
	Then for an invariant density $f_*$ and associated measure $\mu_*$ \cite{Schwegler}, it holds that 
	\begin{equation}
	1\geq \lim_{n \to \infty} \mu_*(\bar{S}_{K, \alpha}^n(C)) \geq \lim_{n \to \infty} \mu_*\left(\bar{S}_{K, \alpha}^n\left\lbrace \bigcup_{j,n'}I_{j,n'}\right\rbrace \right)=1.
	\end{equation}
	Therefore, the map $\bar{S}_{K, \alpha}$ on a phase space $[0, 1)$, is exact.
	Owing to the topological conjugacy, the map $S_{K, \alpha}$ is also exact.
\end{proof}

\section{Scaling behavior}
At first, discuss the case of $\frac{1}{K} < \alpha <1 (\gamma_{K, \alpha}>1)$. 
\begin{lemma}\label{Lemma: scaling alpha=1}
	In the case of $K=2N$, $\gamma_{2N, \alpha}$ behaves as
	\begin{equation}
	\frac{1}{\gamma_{2N, \alpha}} \sim O(\sqrt{1-\alpha}).
	\end{equation}
	in the limit of $\gamma_{2N, \alpha} \to \infty$.
\end{lemma}

\begin{proof}
	In the case of $K=2N$, \eqref{Sinvariant density} is rewritten as
	\begin{equation}
	\alpha = \frac{\coth y}{2N\coth(2Ny)}.
	\end{equation}
	Then, one has that
	\begin{equation}
	\begin{array}{lll}
	1-\alpha &=& \displaystyle\frac{1}{2N}\frac{
		\displaystyle 2N\sum_{k=0}^n {}_{2N}C_{2k}\gamma_{2N, \alpha}^{2N-2k}- \sum_{k=0}^{n-1}{}_{2N}C_{2k+1}\gamma_{2N, \alpha}^{2N-2k}	
	}{
	\displaystyle \sum_{k=0}^{n-1} {}_{2N}C_{2k+1}\gamma_{2N, \alpha}^{2N-2k-1}
}\cdot \frac{
\displaystyle \sum_{k=0}^{n-1} {}_{2N}C_{2k+1}\gamma_{2N, \alpha}^{2N-2k-1}
}{
\displaystyle \sum_{k=0}^{n} {}_{2N}C_{2k}\gamma_{2N, \alpha}^{2N-2k}
}

\\
&=& \displaystyle\frac{1}{2N}\frac{
	\displaystyle 2N\sum_{k=0}^n {}_{2N}C_{2k}\gamma_{2N, \alpha}^{-2k}- \sum_{k=0}^{n-1}{}_{2N}C_{2k+1}\gamma_{2N, \alpha}^{-2k}	
}{
\displaystyle \sum_{k=0}^{n} {}_{2N}C_{2k}\gamma_{2N, \alpha}^{-2k}
},
\end{array}
\end{equation}
In the limit of $\gamma_{2N, \alpha} \to \infty$, it holds that
\begin{equation}
\begin{array}{lll}
1-\alpha &\sim& \displaystyle\frac{1}{2N}\left(2N\cdot {}_{2N}C_{2}\gamma_{2N, \alpha}^{-2}-{}_{2N}C_{3}\gamma_{2N, \alpha}^{-2}\right),\\
\therefore	\displaystyle\frac{1}{\gamma_{2N, \alpha}} &\sim& \sqrt{1-\alpha}.  
\end{array}\label{Scaling 1/K<alpha, K=2N}
\end{equation}
\end{proof}

\begin{lemma}\label{Lemma: scaling alpha=1/K^2}
	In the case of $K=2N+1 $, $\gamma_{2N+1, \alpha}$ behaves as
	\begin{equation}
	\frac{1}{\gamma_{2N+1, \alpha}} \sim O(\sqrt{1-\alpha}).
	\end{equation}
	in the limit of $\gamma_{2N+1, \alpha} \to \infty$.
\end{lemma}

\begin{proof}
	In the case of $K=2N+1$, \eqref{Sinvariant density} is rewritten as
	\begin{equation}
	\alpha = \frac{\coth y}{(2N+1)\coth\left\lbrace (2N+1)y\right\rbrace }.
	\end{equation}
	Then one has that
	\begin{equation}
	\begin{array}{cll}
	1-\alpha &=& \displaystyle  \frac{1}{2N+1}\frac{
		\displaystyle (2N+1)\sum_{k=0}^n {}_{2N+1}C_{2k}\gamma_{2N+1, \alpha}^{-2k}- \sum_{k=0}^{n}{}_{2N+1}C_{2k+1}\gamma_{2N+1, \alpha}^{-2k}	
	}{
	\displaystyle \sum_{k=0}^{n} {}_{2N+1}C_{2k}\gamma_{2N+1, \alpha}^{-2k}
},
\end{array}
\end{equation}
In the limit of $\gamma_{2N+1, \alpha} \to \infty$, it holds that
\begin{equation}
\begin{array}{cll}
1-\alpha &\sim& \displaystyle  \frac{1}{2N+1}\left\lbrace(2N+1)\cdot {}_{2N+1}C_{2}\gamma_{2N+1, \alpha}^{-2}-{}_{2N+1}C_{3}\gamma_{2N+1, \alpha}^{-2}\right\rbrace,\\
\therefore	\displaystyle \frac{1}{\gamma_{2N+1, \alpha}} &\sim& O(\sqrt{1-\alpha}).  
\end{array}\label{Scaling 1/K<alpha, K=2N+1}
\end{equation}
\end{proof}

\begin{lemma}\label{Lemma: scaling alpha=0}
	In the case of $K=2N$, $\gamma_{2N, \alpha}$ behaves as
	\begin{equation}
	\gamma_{2N, \alpha} \sim O(\sqrt{\alpha}).
	\end{equation}
	in the limit of $\gamma_{2N, \alpha} \to 0$.
\end{lemma}

\begin{proof}
	Discuss the case of $K=2N$ and $0<\alpha < \frac{1}{K}$. \eqref{Sinvariant density} is rewritten as
	\begin{equation}
	\alpha = \frac{1}{2N}\tanh y \cdot \tanh (2Ny)
	\end{equation}
	Then one has that
	\begin{equation}
	\alpha = \frac{
		\displaystyle \sum_{k=0}^{N-1} {}_{2N}C_{2k+1}\gamma_{2N, \alpha}^{2k+2}
	}{
	\displaystyle \sum_{k=0}^{N} {}_{2N}C_{2k}\gamma_{2N, \alpha}^{2k}
}.
\end{equation}
In the limit of $\gamma_{2N, \alpha} \to 0$, it holds that
\begin{equation}
\begin{array}{cll}
\alpha &\sim& \frac{1}{2N} \cdot 2N \gamma_{2N, \alpha}^2 = \gamma_{2N, \alpha}^2,\\
\therefore \gamma_{2N, \alpha} &\sim& \sqrt{\alpha}.
\end{array}\label{Scaling 0<alpha<1/K, K=2N}
\end{equation}
\end{proof} 

\begin{lemma}
	In the case of $K=2N+1 $, $\gamma_{2N+1, \alpha}$ behaves as
	\begin{equation}
	\gamma_{2N+1, \alpha} \sim O\left(\sqrt{\alpha-\frac{1}{(2N+1)^2}}\right).
	\end{equation}
	in the limit of $\gamma_{2N+1, \alpha} \to 0$.
\end{lemma}

\begin{proof}
	Discuss the case of $K=2N+1$ and $\frac{1}{K^2}<\alpha < \frac{1}{K}$. \eqref{Sinvariant density} is rewritten as
	\begin{equation}
	\alpha = \frac{\tanh y}{(2N+1)\tanh\left\lbrace (2N+1)y\right\rbrace}.
	\end{equation}
	Then one has that,
	\begin{equation}
	\begin{array}{clll}
	\displaystyle  \alpha - \frac{1}{(2N+1)^2} &=& \frac{
		\displaystyle (2N+1)\tanh y \tanh \{(2N+1)y\}	
	}{
	\displaystyle (2N+1)^2 \tanh \{(2N+1)y\}	
},\\
&=& \frac{
	\displaystyle \sum_{k=0}^{N} \left\lbrace (2N+1){}_{2N+1}C_{2k}- {}_{2N+1}C_{2k+1}\right\rbrace \gamma_{2N+1, \alpha}^{2k}
}{
\displaystyle \sum_{k=0}^{N} {}_{2N+1}C_{2k+1}\gamma_{2N+1, \alpha}^{2k}
},\\
&=& \frac{
	\displaystyle \sum_{k=1}^{N} \left\lbrace (2N+1){}_{2N+1}C_{2k}- {}_{2N+1}C_{2k+1}\right\rbrace \gamma_{2N+1, \alpha}^{2k}
}{
\displaystyle \sum_{k=0}^{N} {}_{2N+1}C_{2k+1}\gamma_{2N+1, \alpha}^{2k}
}.
\end{array}
\end{equation}
In the limit of $\gamma_{2N+1, \alpha} \to 0$, it holds that
\begin{equation}
\begin{array}{cll}
\displaystyle\alpha - \frac{1}{(2N+1)^2} &\sim& \frac{
	\displaystyle \left\lbrace (2N+1){}_{2N+1}C_2 -{}_{2N+1}C_3\right\rbrace \gamma_{2N+1, \alpha}^2	
}{\displaystyle 
2N+1
},\\
\therefore \displaystyle \gamma_{2N, \alpha} &\sim& O\left(\sqrt{\alpha-\frac{1}{(2N+1)^2}}\right).
\end{array}\label{Scaling 1/K^2<alpha<1/K, K=2N+1}
\end{equation}
\end{proof}

From Lemma \ref{Lemma: scaling alpha=1}, \ref{Lemma: scaling alpha=1/K^2} and \ref{Lemma: scaling alpha=0}, 
one knows that there are relations between the parameter $\alpha$ and the scaling parameter $\gamma_{K, \alpha}$.
For all $\alpha$ which satisfied the Condition A, Lyapunov exponent is denoted as
\begin{eqnarray}
\lambda_{K, \alpha} &=& \frac{1}{\pi}\int_{\mathbb{R}\backslash A} \log\left|S_{K, \alpha}'(x)\right| \frac{\gamma_{K, \alpha}}{\gamma_{K, \alpha}^2 + x^2} dx
= \frac{1}{\pi}\int_{-\infty}^\infty \log\left|S_{K, \alpha}'(x)\right| \frac{\gamma_{K, \alpha}}{\gamma_{K, \alpha}^2 + x^2} dx
< \infty, \label{Lyapunov exponent}\\
&=& \left\lbrace
\begin{array}{lll}
\displaystyle	\frac{\gamma_{K, \alpha}}{\pi}\int_{-\infty}^{\infty} \log\left|S_{K, \alpha}'(x)\right| \frac{1}{\gamma_{K, \alpha}^2 + x^2} dx, &\mbox{for}& 0<\gamma_{K, \alpha} \ll 1,\\
\displaystyle	\frac{1}{\gamma_{K, \alpha}\pi}\int_{-\infty}^{\infty} \log\left|S_{K, \alpha}'(x)\right| \frac{1}{1 + (x/\gamma_{K, \alpha})^2} dx, & \mbox{for}& 1\ll \gamma_{K, \alpha}.
\end{array} \label{Lyapunov exponent integral}
\right.
\end{eqnarray}

Define a function $f_1(\theta, \gamma_{K,\alpha})$ to be
\begin{eqnarray}
f_1(\theta, \gamma_{K,\alpha}) = \log \left|\frac{\alpha K^2\sin^2\theta}{\sin^2K\theta}\right|\frac{\gamma_{K,\alpha}}{\gamma_{K,\alpha}^2\sin^2\theta + \cos^2\theta}
\end{eqnarray}
and also define a set of points $\{a_n\}_{n=1}^{K-1}$ such that at $x=a_n\in (0, \pi]$, the function 
$\log \left|\frac{\alpha K^2\sin^2\theta}{\sin^2K\theta}\right|$ is not continuous. 
By changing variable from $x$ to $\cot \theta$, Lyapunov exponent $\lambda_{K,\alpha}$ is rewritten as
\begin{eqnarray}
\lambda_{K,\alpha} = \frac{1}{\pi}\int_{0}^{\pi}f_1(\theta, \gamma_{K,\alpha})d\theta. \label{Lyapunov exponent integral 2}
\end{eqnarray}

\setcounter{theorem}{2}
\renewcommand{\thetheorem}{\Alph{theorem}}
\begin{theorem}
	Suppose that the Condition A is satisfied.
	\begin{itemize}
		\item For any $K\in \mathbb{N}\backslash\{1\}$, it holds that 
		$\nu_1= \frac{1}{2}$ as $\alpha \to 1-0$.
		\item For any $K\in \mathbb{N}\backslash\{1\}$, it holds that 
		$\nu_1= 1$ as $\alpha \to 1+0$.
		\item For $K=2N+1$, it holds that $\nu_2 = \frac{1}{2}$ as $\alpha \to \frac{1}{K^2}+0$.
	\end{itemize}
\end{theorem}
\begin{proof}
	The integrand in \eqref{Lyapunov exponent integral 2} are continuous in $(a_n,a_{n+1}]$ for $0\leq n \leq K$
	where $a_0=0$ and $a_{K} = \pi$.
	The derivative of $f_1(\theta, \gamma_{K,\alpha})$ with respect to $\gamma_{K,\alpha}$ is as follows.
	\begin{eqnarray}
	\frac{\partial f_1}{\partial \gamma_{K,\alpha}} = \frac{1}{\alpha}\frac{\partial \alpha}{\partial \gamma_{K,\alpha}}
	\frac{\gamma_{K,\alpha}}{\gamma_{K,\alpha}^2\sin^2\theta + \cos^2\theta}
	+ \log\left|\alpha\right| \frac{-\gamma_{K,\alpha}^2\sin^2\theta+\cos^2\theta}{(\gamma_{K,\alpha}^2\sin^2\theta+\cos^2\theta)^2}
	+\log \left|\frac{K^2\sin^2\theta}{\sin^2K\theta}\right|\frac{-\gamma_{K,\alpha}^2\sin^2\theta+\cos^2\theta}{(\gamma_{K,\alpha}^2\sin^2\theta+\cos^2\theta)^2}.
	\end{eqnarray}
	The derivative is continuous on each interval $(a_n, a_{n+1}]$.
	
	\noindent (i) In the limit of $\alpha \to 1-0$ $(\gamma_{K,\alpha} \to \infty)$ for any $K$, change variable as $z=\frac{1}{\gamma_{K,\alpha}}$ and define a function $f_2(\theta, z)$ as
	\begin{eqnarray}
	f_2(\theta, z) = \log\left|\frac{\alpha K^2\sin^2\theta}{\sin^2K\theta}\right|\frac{z}{\sin^2\theta+ z^2\cos^2\theta}.
	\end{eqnarray}
	It holds that
	\begin{eqnarray}
	\frac{\partial f_2}{\partial z}(\theta, 0) \sim \log\left|\frac{K^2\sin^2\theta}{\sin^2K\theta}\right|\frac{1}{\sin^2\theta}.
	\end{eqnarray}
	and it dose not depend on $z$. Thus, one has that in the limit of $z \to +0$ $(\gamma_{K, \alpha}\to \infty, \alpha \to 1-0)$, 
	\begin{eqnarray}
	\lambda_{K, \alpha} = \lambda_K(z) &\sim& \lambda(0) + \left[\frac{1}{\pi}\int_0^\pi\log\left|\frac{K^2\sin^2\theta}{\sin^2K\theta}\right|\frac{1}{\sin^2\theta}d\theta\right]z + O(z^2).
	\end{eqnarray}
	where $\lambda_K(z=0)=0$. Therefore it holds that
	\begin{eqnarray}
	\therefore \lambda_{K, \alpha} &\sim& O\left(z\right) \sim O\left(\sqrt{1-\alpha}\right). \label{asymptotic critical exponent infty}
	\end{eqnarray}
	
	\noindent (ii) In the case of $K=2N+1$, there is a relation as $\alpha \sim \gamma_{K,\alpha}^2+\frac{1}{K^2}$
	in the limit of $\gamma_{K,\alpha}\to +0$ $(\alpha \to \frac{1}{K^2}+0)$.
	Then it holds that
	\begin{eqnarray}
	\frac{\partial f_1}{\partial \gamma_{K,\alpha}}(\theta, 0) \sim 
	\log \left|\frac{\sin^2\theta}{\sin^2K\theta}\right|\frac{1}{\cos^2\theta}, ~K=2N+1
	\end{eqnarray}
	and it dose not depend on $\gamma_{K,\alpha}$. Thus by a Taylor expansion with respect to $\gamma_{K,\alpha}$, 
	one has that in the limit of $\gamma_{K, \alpha} \to 0$,
	\begin{eqnarray}
	\lambda_{K,\alpha} =\lambda_K(\gamma_{K,\alpha}) \sim
	\lambda_K(0) +  \left[\displaystyle \frac{1}{\pi}\int_0^\pi \left\lbrace\log \left|\frac{\sin^2\theta}{\sin^2K\theta}\right|\frac{1}{\cos^2\theta}\right\rbrace d\theta\right] \gamma_{K,\alpha} +O(\gamma_{K, \alpha}^2), ~K=2N+1.
	\end{eqnarray}
	where $\lambda_K(\gamma_{K,\alpha}=0)=0$. Therefore it holds that
	\begin{eqnarray}
	\lambda_{K,\alpha} \sim O(\gamma_{K,\alpha}) \sim 
	O\left(\sqrt{\alpha-\frac{1}{K^2}}\right),~ K=2N+1. \label{asymptotic critical exponent 0}
	\end{eqnarray}
	in the limit of $\alpha \to \frac{1}{K^2}+0$.
	
	\noindent (iii) In the limit of $\alpha\to 1+0$ for any $K$, Lyapunov exponent is denoted as for $\alpha >1$, 
	\begin{equation}
	\begin{array}{lll}
	\lambda_{K, \alpha} &=& \log \alpha = \log \left\lbrace 1+(\alpha-1)\right\rbrace,\\
	&=& (\alpha-1) -\frac{1}{2}(\alpha-1)^2 + \frac{1}{3}(\alpha-1)^3-\cdots.
	\end{array}
	\end{equation}
	Therefore, it holds that
	\begin{equation}
	\lambda_{K, \alpha} \sim O(\alpha-1). \label{asymptotic critical exponent 1+0}
	\end{equation}
	From equations \eqref{asymptotic critical exponent 0}, \eqref{asymptotic critical exponent infty} and
	\eqref{asymptotic critical exponent 1+0}, one sees 
	that for any $K$, the critical exponent of Lyapunov exponent $\nu_1$ is $\frac{1}{2}$ in the limit of $\alpha\to 1-0$, 
	that for $K=2N+1$, $\nu_2=\frac{1}{2}$ in the limit of $\alpha \to \frac{1}{K^2}+0$ and
	that for any $K$, $\nu_1=1$ in the limit of $\alpha \to 1+0$.
\end{proof}

From above discussion, it is proven that the derivatives of Lyapunov exponent with respect to parameter $\alpha$ diverge at critical points, which means that
the parameter dependence of Lyapunov exponent diverges at critical points. This result implies that difficulty of 
calculating Lyapunov exponent near the critical points. 

\subsection{Scaling behavior for $K=3,4$ and 5}
The solutions of \eqref{Sinvariant density} which satisfied $0<\gamma_{K, \alpha} <\infty$
in the cases of $K=3, 4$ and 5 are determined uniquely when the Condition A is satisfied as follows.
\begin{equation}
\begin{array}{lll}
\gamma_{3, \alpha} &=& \displaystyle\sqrt{\frac{9\alpha-1}{3- 3\alpha}},\\
\gamma_{4, \alpha} &=& \displaystyle \sqrt{\frac{6\alpha-1+\sqrt{32\alpha^2-8\alpha+1}}{2(1-\alpha)}},\\
\gamma_{5, \alpha} &=& \displaystyle\sqrt{\frac{-5(1-5\alpha)+\sqrt{20(25\alpha^2-6\alpha+1)}}{5(1-\alpha)}}.
\end{array}
\end{equation}
From \eqref{Lyapunov exponent}, one has analytical formulae of Lyapunov exponent as follows.
\begin{equation}
\begin{array}{lll}
\lambda_{3, \alpha} &=&  \displaystyle \log \left|\frac{1}{\alpha}\left( \frac{3(1-\alpha)}{8}\right)^2\left[1+\sqrt{\frac{9\alpha-1}{3-3\alpha}}\right]^4 \right|,
\\
\lambda_{4, \alpha} &=& \displaystyle \log\left|\frac{\alpha(1+\gamma_4)^6}{\gamma_4^2(1+\gamma_4^2)^2}\right|,\\
\lambda_{5,\alpha} &=& \displaystyle \log \left|\frac{25}{256\alpha}\frac{(1-\alpha)^4}{(\sqrt{125\alpha^2-30\alpha+5}+11\alpha-1)^2}|1+\gamma_5|^8\right|,
\end{array} \label{SLyapunov K}
\end{equation}

In the case of $K=3$, equation \eqref{SLyapunov K} converges to zero in the limit of 
$\alpha \to \frac{1}{9}+0$ and $\alpha \to 1-0$, and the derivative of Lyapunov exponent 
$\frac{\partial \lambda_{3, \alpha}}{\partial \alpha}$ diverges at $\alpha = \frac{1}{9}, 1$.
When the parameter $\alpha$ is close to $\frac{1}{9}$, the Lyapunov exponent grows as follows.
\begin{eqnarray}
\lambda_{3, \alpha} &=& -\log\left|1+9\left(\alpha-\frac{1}{9}\right)\right|
+ 2\log\left|1-\frac{9}{8}\left(\alpha-\frac{1}{9}\right)\right|
+4\log\left|1+ \sqrt{\frac{3\left(\alpha- \frac{1}{9}\right)}{1-\alpha}}\right|,\nonumber\\
&\simeq& -9\left(\alpha-\frac{1}{9}\right) -\frac{9}{4}\left(\alpha- \frac{1}{9}\right)
+4 \sqrt{\frac{3\left(\alpha-\frac{1}{9}\right)}{1-\alpha}},\nonumber\\
&\simeq& 6\sqrt{\frac{3}{2}}\sqrt{\alpha-\frac{1}{9}}.
\end{eqnarray}
Then, the critical exponent $\nu_2$ of Lyapunov exponent at $\alpha= \frac{1}{9}$ is $\frac{1}{2}$.
In the case of $\alpha \lesssim 1$, Lyapunov exponent $\lambda_{3, \alpha}$ behaves as follows.
\begin{eqnarray}
\lambda_{3, \alpha} &=& 2\log\left|1+\frac{9}{8}\left(\alpha-1\right)\right|
-\log\left|1+\left(\alpha-1\right)\right|+ 4\log\left|1+ \sqrt{\frac{1-\alpha}{3\left(\alpha-\frac{1}{9}\right)}}\right|,\nonumber\\
&\simeq& -\frac{9}{4}(\alpha-1)+ (\alpha-1)+ 4\sqrt{\frac{1-\alpha}{3\left(\alpha-\frac{1}{9}\right)}},\nonumber\\
&\simeq& 2\sqrt{\frac{3}{2}}\sqrt{1-\alpha}.
\end{eqnarray}
Then, at $\alpha=1$ one has $\nu_1 = \frac{1}{2}$. 

For $K=3$ and for $\frac{1}{9}< \alpha < 1$, the fixed point is only 
\begin{eqnarray}
x^*= 0.
\end{eqnarray}
Then Floquet multiplier $\chi$ at $\alpha = \frac{1}{9}$ and $\alpha=1$ are denoted as follows.
\begin{equation}
\begin{array}{lllll}
\chi_{3, \frac{1}{9}} &=& S_{3, \frac{1}{9}}'(0) &=&1,\\
\chi_{3, 1} &=& S_{3, 1}'(0) &=& 1.
\end{array}
\end{equation}
From these results, we can say that only \textit{Type 1} intermittency occurs for $K=3$.
These results are new phenomena since for the Generalized Boole transformation, one can observe
two different intermittent type, \textit{Type 1} and \textit{Type 3} \cite{Okubo}.

In the case of $K=4$, Lyapunov exponent \eqref{SLyapunov K} is denoted as
\begin{eqnarray}
\lambda_{4, \alpha} &=& \log\left|\frac{\alpha}{\gamma_4^2}\right| + 6\log|1+\gamma_4| -2\log|1+\gamma_4^2|. \label{Lyapunov K=4 divide}
\end{eqnarray}
Let us discuss the scaling behavior of $\lambda_{4, \alpha}$ at $\alpha=0$.
Now the first term of \eqref{Lyapunov K=4 divide} is rewritten as
\begin{eqnarray}
\log\left|\frac{\alpha}{\gamma_4^2}\right|  &=& \log\left|\frac{2\alpha(1-\alpha)}{6\alpha-1+\sqrt{32\alpha^2-8\alpha+1}}\right|,\nonumber\\
&=& \log\left|\frac{(1-\alpha)\left\lbrace \sqrt{32\alpha^2-8\alpha+1}-6\alpha +1\right\rbrace }{-2\alpha+2}\right|,\label{Lyapunov K=4 first term}\\
\therefore \lim_{\alpha \to 0} \log\left|\frac{\alpha}{\gamma_4^2}\right|
&=& \frac{1\times(\sqrt{1}-0+1)}{2} =1.
\end{eqnarray}
Then, 
\begin{eqnarray}
\gamma_4 &\simeq& \sqrt{-1+1+\frac{1}{2}\left(32\alpha^2-8\alpha\right)+6\alpha} \simeq \sqrt{2\alpha}, \nonumber\\
\Longrightarrow \log\left|1+\gamma_4\right| &\simeq& \sqrt{2\alpha} ~\mbox{and}~\log\left|1+\gamma_4^2\right| \simeq 2\alpha, \nonumber\\
\therefore \lambda_{4, \alpha} &\simeq& \sqrt{2\alpha}.
\end{eqnarray}
Therefore, the critical exponent of Lyapunov exponent for $K=4, \alpha=0$ is
\begin{equation}
\nu_{3} = \frac{1}{2}.
\end{equation}

\noindent Next, consider the scaling behavior of $\lambda_{4, \alpha}$ at $\alpha=1-0$.
From \eqref{Lyapunov K=4 first term}, it holds that  near $\alpha= 1$,
\begin{eqnarray}
\log\left|\frac{\alpha}{\gamma_4^2}\right|&=&\log\left|\frac{(1-\alpha)\left\lbrace \sqrt{32\alpha^2-8\alpha+1}-6\alpha +1\right\rbrace }{-2\alpha+2}\right|,\nonumber\\
&=& \log \left|\frac{5\sqrt{32(\alpha-1)^2+56(\alpha-1)+25}-6(\alpha-1)-5}{2}\right|,\nonumber\\
&\simeq& \log\left|\frac{\frac{16}{5}(\alpha-1)^2 -\frac{2}{5}(\alpha-1)}{2}\right|,\nonumber\\
&\simeq& \log\left|\frac{1-\alpha}{5}\right| \simeq \alpha.
\end{eqnarray}
The second and third terms of \eqref{Lyapunov K=4 divide} are denoted as
\begin{eqnarray}
\log\left|\frac{(1+\gamma_4)^6}{(1+\gamma_4^2)^2}\right| &=& 
6\log\left| \sqrt{2(1-\alpha)} + \sqrt{6\alpha-1+\sqrt{32\alpha^2-8\alpha+1}}\right|\nonumber\\
& &-\log 2(1-\alpha) -2\log\left|1+4\alpha+\sqrt{32\alpha^2-8\alpha+1}\right|,\nonumber\\
&\simeq&  \log 5 + 6\log \left(1+\sqrt{\frac{1-\alpha}{5}}\right) - \log(1-\alpha),\nonumber\\
\therefore \lambda_{4, \alpha} &\simeq& \frac{6\sqrt{5}}{5}\sqrt{1-\alpha}, ~\mbox{for}~ \alpha \simeq 1.
\end{eqnarray}
Therefore, the critical exponent of Lyapunov exponent for $K=4, \alpha=1-0$ is
\begin{equation}
\nu_{1} = \frac{1}{2}.
\end{equation}
In the case of $K=4$, fixed points of $S_{4, \alpha}$ are as follows.
\begin{equation}
x_{4*} = \left\lbrace
\begin{array}{cl}
0, & \alpha = 0,\\
\pm \sqrt{\frac{1-6\alpha+\sqrt{40\alpha^2-16\alpha +1}}{1-\alpha}}, &0<\alpha<1,\\
\pm \frac{1}{\sqrt{5}}, & \alpha = 1.
\end{array}
\right.
\end{equation}
In order to obtain the Floquet multiplier at $\alpha =0$, apply scale transformation such that $x = \sqrt{\alpha} y$. Then, one obtains 
following equations as
\begin{equation}
\begin{array}{lll}
y_{n+1} &=& \widehat{S}_{4, \alpha}(y_n) = \frac{\alpha^2y_n^4 -6\alpha y_n^2 +1}{\alpha y_n^2 -y_n},\\
\widehat{S}_{4, 0}(y_n) &=& -\frac{1}{y_n}.
\end{array}
\end{equation}
Then $y = \pm i$ are the fixed points for $\widehat{S}_{4, 0}$ and one has that 
\begin{equation}
\begin{array}{lll}
\widehat{S}_{4, 0}'(y_n) &=& \frac{1}{y_n^2},\\
\widehat{S}_{4, 0}'(\pm i) &=& -1.
\end{array} \label{SFloquet multiplier K=4}
\end{equation}
Thus, the Floquet multiplier at $\alpha = 0$ is -1.

The values of $S_{4, \alpha}'$ at the other fixed points are as follows.
\begin{equation}
\begin{array}{ccl}
\displaystyle \lim_{\alpha \to +0} S_{4, \alpha}'\left(\pm \sqrt{\frac{1-6\alpha+\sqrt{40\alpha^2-16\alpha +1}}{1-\alpha}}\right) &=& 0,\\
S_{4, 1}'\left(\pm \frac{1}{\sqrt{5}}\right) &=& \frac{27}{2},\\
S_{4, 1}'\left(\pm \infty\right) &=& 1.
\end{array}
\label{Sderivative K=4}
\end{equation}
For $K=4$, from \eqref{SFloquet multiplier K=4} and \eqref{Sderivative K=4}, 
one obtains the Floquet multiplier at $\alpha=0$ and $\alpha=1$ as 
\begin{eqnarray}
\begin{array}{lll}
\chi_{4, 0} &=& -1,\\
\chi_{4, 1} &=& 1.
\end{array}
\end{eqnarray}
This result indicates that \textit{Type 3} intermittency occurs at $\alpha=0$ and
\textit{Type 1} intermittency occurs at $\alpha=1$.
These results are consistent with that of Generalized Boole transformation \cite{Okubo}.

In the case of $K=5$, equation \eqref{SLyapunov K} converges to zero in the limit of 
$\alpha \to \frac{1}{25}+0$ and $\alpha \to 1-0$. In addition, it holds that
$\frac{\partial \lambda_{5, \alpha}}{\partial \alpha}\left(\frac{1}{25}\right)=\frac{\partial \lambda_{5, \alpha}}{\partial \alpha}(0) = \infty$.
Lyapunov exponent $\lambda_{5, \alpha}$ for $\frac{1}{25}< \alpha <1$ is divided as
\begin{equation}
\lambda_{5, \alpha} = \log\frac{25}{256} - \log|\alpha| + 4\log|1-\alpha|
-2\log\left|1-11\alpha -\sqrt{125\alpha^2 -30\alpha +5}\right|
+ 8\log|1+ \gamma_5^*|.\label{Lyapunov K=5 divide}
\end{equation}
The forth term of \eqref{Lyapunov K=5 divide} at $\alpha = \frac{1}{25}$ is denoted as
\begin{eqnarray}
\log\left|-1-\frac{11}{25}\right|> \log|1|=0
\end{eqnarray}
and denoted at $\alpha=1$ as
\begin{eqnarray}
\log|-20| > \log|1|=0.
\end{eqnarray}
Then the the forth term is not dominant near $\alpha= \frac{1}{25}$ and $\alpha=1$. 
Thus in considering the scaling behavior, we do not have to care the forth term.

\noindent When the parameter $\alpha$ is close to $\frac{1}{25}$, the Lyapunov exponent grows as follows
\begin{eqnarray}
\lambda_{5,\alpha} \simeq \gamma_5^* &=& \sqrt{\frac{-5(1-5\alpha)+ \sqrt{20(25\alpha^2-6\alpha +1)}}{5(1-\alpha)}},\\
&=& \sqrt{\frac{4\left[\sqrt{1+20 \left\lbrace 25\left(\alpha - \frac{1}{25}\right)^2 -4 \left(\alpha- \frac{1}{25}\right)\right\rbrace }
		+\frac{25}{4}\left(\alpha - \frac{1}{25}\right) -1 \right]}{5(1-\alpha)}},\nonumber\\
&\simeq& \frac{185}{4}\sqrt{\alpha - \frac{1}{25}}.
\end{eqnarray}
Therefore, the critical exponent $\nu_{2}$ is denoted as
\begin{eqnarray}
\nu_{2} = \frac{1}{2}.
\end{eqnarray} 

\noindent For $\alpha \lesssim 1$, Lyapunov exponent grows as follows.
\begin{eqnarray}
\lambda_{5, \alpha} &\simeq& 4\log\left|(1-\alpha) (1+ \gamma_5^*)^2\right|,\nonumber\\
&\simeq& 4\log \left|8+\sqrt{1-\alpha}\right|,\nonumber\\
&\simeq& \frac{\sqrt{1-\alpha}}{2}.
\end{eqnarray}
Therefore, the critical exponent $\nu_{1}$ is denoted as
\begin{eqnarray}
\nu_{1} = \frac{1}{2}.
\end{eqnarray} 
In the case of $K=5$, fixed points of $S_{5, \alpha}$ are as follows.
\begin{equation}
x_* = \left\lbrace
\begin{array}{ll}
0, & 0< \alpha \leq 1,\\
\pm \sqrt{\frac{5}{3}}, & \alpha = \frac{1}{25},\\
\pm \sqrt{\frac{5(1-5\alpha)+ 2\sqrt{5(25\alpha^2-6\alpha +1)}}{5(1-\alpha)}}, & \frac{1}{25} < \alpha <1,\\
\pm \sqrt{\frac{3}{5}}, & \alpha =1.
\end{array}
\right.
\end{equation}
At fixed points $x_* = 0, \pm \sqrt{\frac{5}{3}}$ and $\pm \sqrt{\frac{3}{5}}$ , the derivatives 
$S_{5, \alpha}' (x) = \frac{25\alpha(1+x^2)^4}{(5x^4-10x^2+1)^2}$ are
\begin{equation}
\begin{array}{lcl}
S_{5, \frac{1}{25}}'\left(0\right) &=& 1,\\
S_{5, 1}' (0) &=& 25,\\
S_{5, \frac{1}{25}}'\left(\pm \sqrt{\frac{5}{3}}\right) &=& 16,\\
S_{5, 1}'\left(\pm \sqrt{\frac{3}{5}}\right) &=& 16,\\
S_{5, 1}'\left(\pm \infty\right) &=& 1,
\end{array}  \label{Sderivative K=5}
\end{equation} 
From \eqref{Sderivative K=5}, for $K=5$ and at $\alpha = \frac{1}{25}$ and $\alpha =1$, 
one obtains the Floquet multipliers as
\begin{equation}
\begin{array}{lllll}
\chi_{5, \frac{1}{25}} &=& S_{5, \frac{1}{25}}'(0) &=& 1,\\
\chi_{5, 1} &=& S_{5, 1}'(0) &=& 1.
\end{array}
\end{equation}
Therefore, similar to the case of $K=3$, only \textit{Type 1} intermittency occurs, which 
is different from the case with Generalized Boole transformation and the case of $K=4$.

\end{document}